\documentclass[sn-mathphys,Numbered]{sn-jnl}

\usepackage{multirow,xcolor,textcomp,manyfoot,booktabs,algpseudocode,listings}%
\usepackage[title]{appendix}%
\usepackage{amsmath,amsthm,amsfonts,amssymb,geometry,subcaption}
\usepackage{graphicx,float,algorithm,algorithmicx,multirow,mathrsfs}
\usepackage{braket}

\theoremstyle{thmstyleone}%
\newtheorem{theorem}{Theorem}
\newtheorem{lemma}[theorem]{Lemma}
\newtheorem{corollary}[theorem]{Corollary}

\theoremstyle{thmstyletwo}%

\theoremstyle{thmstylethree}%
\newtheorem{definition}{Definition}%

\begin{document}
	
	\title{A Logical Formalism of Hardy-type Paradox}
	
	\author[1]{\fnm{Songyi} \sur{Liu}}\email{liusongyi@buaa.edu.cn}
	
	\author*[1]{\fnm{Yongjun} \sur{Wang}}\email{wangyj@buaa.edu.cn}
	
	\author[1]{\fnm{Baoshan} \sur{Wang}}\email{bwang@buaa.edu.cn}
	
	\author[1]{\fnm{Chang} \sur{He}}\email{hechang@buaa.edu.cn}
	
	\author[1]{\fnm{Yunyi} \sur{Jia}}\email{by2309005@buaa.edu.cn}
	
	\affil*[1]{\orgdiv{School of Mathematical Sciences}, \orgname{Beihang University}, \orgaddress{ \city{Beijing}, \postcode{100191}, \country{China}}}
	
	\abstract{Hardy-type paradoxes provide elegant, inequality-free proofs of quantum contextuality. We introduce a unified logical formalism for these paradoxes, termed logical Hardy-type paradoxes. For any finite quantum scenario of ideal measurements, we prove that the existence of a logical Hardy-type paradox is equivalent to logical contextuality. Specifically, strong contextuality is equivalent to logical Hardy-type paradoxes with success probability $\mathrm{SP} = 1$. These results generalize prior work on $(2,k,2)$, $(2,2,d)$, and $n$-cycle scenarios. We analyze logical Hardy-type paradoxes in the Mansfield and Klyachko-Can-Binicio\u{g}lu-Shumovsky (KCBS) scenarios. In the KCBS scenario, we show that there is exactly one type of logical Hardy-type paradox, achieving $\mathrm{SP}\approx 10.56\%$ for a specific parameter setting.}
	
	
	
	
	\maketitle
	
	\section*{Declarations}
	
	\noindent\textbf{Competing interests.}
	The authors have no relevant financial or non-financial interests to disclose.
	
	\vspace{10pt}
	
	\noindent\textbf{Acknowledgments.}
	The work was supported by National Natural Science Foundation of China (Grant No. 12371016, 11871083) and National Key R\&D Program of China (Grant No. 2020YFE0204200).
	
	\section{Introduction}
	
	The Bell-Kochen-Specker (BKS) theorem \cite{Kochen1967The} establishes quantum contextuality, demonstrating the incompatibility of quantum mechanics with hidden-variable theories \cite{Budroni2022Kochen}. Bell nonlocality \cite{Brunner2014Bell} is a special case of contextuality for spacelike-separated systems, which is also useful to interpret semantic paradoxes \cite{Zhou2024Quantum}. While early proofs of nonlocality relied on inequalities \cite{Bell1964On, Clauser1969Proposed}, Hardy introduced an elegant, inequality-free proof using a logical paradox. Hardy-type paradoxes exploit quantum-realizable logical contradictions that are classically impossible \cite{Hardy1992Quantum,Hardy1993Nonlocality}. This formulation is considered one of the simplest proofs of Bell nonlocality \cite{Mermin1994Quantum} and has been verified in numerous experiments \cite{Boschi1997Ladder,Irvine2005Realization,Barbieri2005Test,Marques2014Experimental,Luo2018Experimental,Hou2021Quantum}.
	
	The original Hardy paradox, set in the $(2,2,2)$ Bell scenario, achieves a maximum success probability of $\mathrm{SP}_{\max} \approx 9\%$ \cite{Hardy1993Nonlocality}. Subsequent generalizations have extended it to broader scenarios. For $(n,2,2)$ Bell scenarios involving Greenberger-Horne-Zeilinger (GHZ) states, the maximum success probability is $12.5\%$ for $n=3$ \cite{Cereceda2004Hardy}, approaching approximately $15.6\%$ asymptotically \cite{Minh2023Increased}. Further work has extended Hardy-type paradoxes to $(2,k,2)$ and $(2,2,d)$ scenarios \cite{Boschi1997Ladder,Mansfield2012Hardy,Chen2013Hardy}, culminating in a unified treatment for $(2,k,d)$ scenarios \cite{Meng2018Hardy}. This yields $\mathrm{SP}_{\max} \approx 40.2\%$ for the $(2,5,3)$ case, surpassing earlier bounds \cite{Cabello1998Ladder}. Beyond Bell scenarios, contextuality theory provides a framework for further generalization \cite{Budroni2022Kochen}. For instance, Hardy-type paradoxes have been formulated for $n$-cycle scenarios using exclusivity graphs \cite{Cabello2013Simple}, with the maximum success probability for any 5-cycle scenario shown to be $SP_{\max} = 1/9 \approx 11.1\%$ \cite{Santos2021Conditions}.
	
	Despite numerous known examples, a unified mathematical formalization of Hardy-type paradoxes for general quantum scenarios remains lacking. For instance, although a generalized Hardy-type paradox has been presented for the $5$-cycle scenario, specific cases such as the Klyachko–Can–Binicio\u{g}lu–Shumovsky (KCBS) scenario \cite{Alexander2008Simple}, which is the simplest $5$-cycle scenario, have not been analyzed. In this work, we show that the KCBS scenario admits exactly one type of Hardy-type paradox.
	
	Mansfield et al. proved that the existence of Hardy-type paradoxes on $(2,k,2)$ and $(2,2,d)$ Bell scenarios is equivalent to logical contextuality \cite{Mansfield2012Hardy}, a notion introduced within the sheaf-theoretic approach \cite{Abramsky2011sheaf}. This equivalence was later extended to $n$-cycle scenarios \cite{Santos2021Conditions}. Since logical contextuality can be verified algorithmically, this equivalence provides a powerful tool for identifying Hardy-type paradoxes. However, \cite{Mansfield2013The} also exhibited a logically contextual state on the $(2,3,3)$ scenario that does not witness any coarse-grained Hardy-type paradox, suggesting the equivalence may not hold in general. Nevertheless, for ideal measurements, we prove that a more general notion, termed the logical Hardy-type paradox, is precisely equivalent to logical contextuality on any finite quantum scenario.
	
	We work within the framework of exclusive partial Boolean algebras (epBAs) \cite{Kochen1967The, Abramsky2020The, Liu2025Atom}, an extension of standard quantum logic \cite{Birkhoff1936The}. An epBA describes a scenario satisfying Specker’s principle \cite{Specker1960Die}, the no-signaling \cite{Popescu1994Quantum} and the logical exclusivity principle \cite{Abramsky2020The}, which held by quantum scenarios of ideal measurements. Its key idea is to treat measurement events as fundamental, and binary operations are defined only for compatible events. In the following sections, we introduce a unified formalism called logical Hardy-type paradoxes within the epBA framework. We establish their equivalence to logical contextuality, and illustrate the results with analyses of the Mansfield, Bell, and KCBS scenarios. This paper is restricted to ideal measurements scenarios.
	
	\section{Preliminaries}\label{sec-pre}
	
	This section introduces the necessary concepts of partial Boolean algebras in the field of quantum contextuality. For further details, see \cite{Kochen1967The,Van2012Noncommutativity,Abramsky2020The,Liu2025Atom,Liu2025The}.
	
	An experimental setup consists of two components: a scenario and a state. The scenario describes the structure of measurements and observable events. If the experiment satisfies no-signaling and Specker's principle (i.e., pairwise compatibility implies global compatibility), as is the case for ideal quantum measurements, its scenario can be described by a partial Boolean algebra (pBA), defined as follows.
	
	\begin{definition}
		A \textbf{partial Boolean algebra} (pBA) is a structure $(\mathcal{A},\odot,\lnot,\land, 0, 1)$ where:
		\begin{enumerate}
			\item $\mathcal{A}$ is a set;
			\item $\odot\subseteq\mathcal{A}\times\mathcal{A}$ is a reflexive, symmetric binary relation (compatibility);
			\item $\lnot:\mathcal{A}\rightarrow\mathcal{A}$ is a total unary operation;
			\item $\land:\odot\rightarrow\mathcal{A}$ is a partial binary operation;
			\item $0, 1 \in\mathcal{A}$ are the bottom and top elements,
		\end{enumerate}
		such that every pairwise-compatible subset $S \subseteq\mathcal{A}$ (i.e., $a \odot b$ for all $a, b \in S$) is contained in a Boolean subalgebra $\mathcal{B} \subseteq\mathcal{A}$, where the operations of $\mathcal{B}$ are the restrictions of $\lnot$ and $\land$.
	\end{definition}
	
	A state on a pBA assigns probabilities to events.
	
	\begin{definition}\label{def-pba_state}
		Let $\mathcal{A}$ be a pBA. A \textbf{state} on $\mathcal{A}$ is a function $p: \mathcal{A} \to [0, 1]$ satisfying, for all $a,b\in\mathcal{A}$:
		\begin{enumerate}
			\item $p(0) = 0$ and $p(1) = 1$;
			\item $p(\neg a) = 1 - p(a)$;
			\item $p(a) + p(b) = p(a \land b) + p(a \lor b)$ if $a \odot b$.
		\end{enumerate}
		A state $p$ is \textbf{deterministic} if $p(a)\in\{0,1\}$ for all $a \in\mathcal{A}$. The set of all states on $\mathcal{A}$ is denoted $s(\mathcal{A})$; the set of deterministic states is $s_d(\mathcal{A})$.
	\end{definition}
	
	Thus, a pair $(\mathcal{A}, p)$ fully describes a general static experiment.
	
	Define the partial order $\leq$ and the exclusivity relation $\perp$ on a pBA $\mathcal{A}$ by:
	\begin{equation}
		\begin{split}
			a \leq b &\quad\text{if}\quad a \land b = a,\\
			a \perp b &\quad\text{if}\quad a \leq c \text{ and } b \leq \neg c \text{ for some } c \in \mathcal{A}.
		\end{split}
	\end{equation}
	
	\begin{definition}\label{def-LEP}
		A pBA $\mathcal{A}$ satisfies the \textbf{logical exclusivity principle (LEP)} if every pair of exclusive events is compatible, i.e., ${\perp} \subseteq {\odot}$. A pBA satisfying LEP is called an \textbf{exclusive partial Boolean algebra (epBA)}.
	\end{definition}
	
	LEP implies several useful properties \cite{Abramsky2020The}: it entails the probability exclusivity principle (PEP) \cite{Fritz2013Local,Adan2014Graph} and is equivalent to transitivity. An important motivation for LEP is the following result. For a finite pBA $\mathcal{A}$, let a simple graph $\mathcal{G}_a(\mathcal{A})$ be its \textbf{atom graph}, whose vertices are the atoms of $\mathcal{A}$ and edges represent compatibility. A function $s:\mathcal{G}_a(\mathcal{A})\to [0,1]$ is a \textbf{state} on $\mathcal{G}_a(\mathcal{A})$ if $\sum_{v\in C} s(v)=1$ for every maximal clique $C \subseteq \mathcal{G}_a(\mathcal{A})$.
	
	\begin{theorem}[\cite{Liu2025Atom}]\label{thm:atom_graph_correspondence}
		Let $\mathcal{A}$ and $\mathcal{A}'$ be finite epBAs. Then $\mathcal{A}\cong\mathcal{A}'$ if and only if $\mathcal{G}_a(\mathcal{A}) \cong \mathcal{G}_a(\mathcal{A}')$. Moreover, every state on $\mathcal{G}_a(\mathcal{A})$ extends uniquely to a state on $\mathcal{A}$.
	\end{theorem}
	
	Thus, a finite epBA is completely determined by its atom graph, which simplifies many questions.
	
	\begin{definition}
		A \textbf{(finite) general system} is a pair $(\mathcal{A}, p)$, where $\mathcal{A}$ is a (finite) epBA and $p \in s(\mathcal{A})$.
	\end{definition}
	
	In quantum mechanics, observable events correspond to projectors. For projectors $\hat{P}, \hat{Q}$ on a Hilbert space $\mathcal{H}$, define:
	\begin{itemize}
		\item $\hat{P} \odot \hat{Q}$ iff $\hat{P}\hat{Q} = \hat{Q}\hat{P}$ (compatibility);
		\item $\lnot\hat{P} := \mathbf{I} - \hat{P}$ (orthogonal complement);
		\item $\hat{P} \land \hat{Q} := \hat{P}\hat{Q}$ when $\hat{P} \odot \hat{Q}$ (intersection);
		\item $\mathbf{I}$ and $\mathbf{0}$ denote the identity and zero operators.
	\end{itemize}
	These operations endow the set of projectors $\mathbf{P}(\mathcal{H})$ with the structure of an epBA, denoted $(\mathbf{P}(\mathcal{H}),\odot,\neg,\land,\mathbf{0},\mathbf{I})$. Subalgebras of $\mathbf{P}(\mathcal{H})$ are called \emph{quantum scenarios}. A quantum scenario $\mathcal{Q}$ together with a density operator $\rho$ describes a static quantum experiment.
	
	\begin{definition}
		A \textbf{quantum system} is a pair $(\mathcal{Q},p)$, where $\mathcal{Q}\subseteq\mathbf{P}(\mathcal{H})$ is a quantum scenario and $p$ is the state induced by $\rho$ via $p(\hat{P})=\operatorname{Tr}(\rho\hat{P})$.
	\end{definition}
	
	For simplicity, we use the same symbol $\rho$ for the density operator and the induced state. Any quantum system is a general system \cite{Liu2025The}.
	
	In an epBA $\mathcal{A}$, for compatible events $a \odot b$, the disjunction is defined via De Morgan's law:
	\begin{equation}
		a \lor b :=\lnot(\lnot a \land\lnot b).
	\end{equation}
	If $\mathcal{A}$ embeds into a Boolean algebra, there exists a canonical embedding into $\mathcal{A}^c :=\mathcal{P}(s_d(\mathcal{A}))$ \cite{Liu2025The}:
	\begin{equation}
		\begin{split}
		\label{eq-classical-embedding}
		i_{\mathcal{A}}:\mathcal{A}&\to\mathcal{A}^c\\
		e&\mapsto e^c:=\{\lambda\in s_d(\mathcal{A}): \lambda(e)=1\}.
		\end{split}
	\end{equation}
	This embedding facilitates comparison between quantum and classical logic. For consistency, we use logical notation in $\mathcal{A}^c$: for $E, F \subseteq s_d(\mathcal{A})$:
	\begin{equation}
		\neg E := s_d(\mathcal{A})\setminus E;\quad
		E \land F := E \cap F;\quad
		E \lor F := E \cup F;\quad
		E\leq F \;\Leftrightarrow\; E\subseteq F.
	\end{equation}
	For notational convenience, the logical bottom element is uniformly denoted $\bot$ (representing $0_{\mathcal{A}}$ in $\mathcal{A}$, $\mathbf{0}$ in quantum scenarios, or $\emptyset$ in $\mathcal{A}^c$).
	
	\section{Logical Hardy-type paradox}\label{sec-LHp}
	
	Hardy-type paradoxes provide inequality-free proofs of quantum contextuality. They are characterized by a set of conditions that logically imply a certain event; a contradiction arises when this event is experimentally observed. For example, in the $(2,2,2)$ Bell scenario, a Hardy-type paradox is given by the probability constraints \cite{Chen2024Hardy}:
	\begin{equation}
		\label{original Hardy paradox}
		\begin{split}
			P(0,0|0,0) = 0&,\quad P(1,1|0,1) = 0, \\
			P(1,1|1,0) = 0&,\quad P(1,1|1,1) = q>0.
		\end{split}
	\end{equation}
	Here, $P(x,y|i,j)$ is the joint probability of outcomes $x$ (Alice) and $y$ (Bob) given measurement settings $i$ and $j$. Let $A_i$ and $B_j$ denote the corresponding observables.
	
	In a classical (noncontextual) model, the first three constraints force $P(1,1|1,1)=0$. Hence, observing $q>0$ in a quantum experiment reveals contextuality. The value $q$ is called the \textbf{success probability (SP)} of the paradox.
	
	A zero probability event $e$ can be treated as logically false, so $\neg e$ is true. Defining $a_0,a_1,b_0,b_1$ as the events $A_0=1$, $A_1=1$, $B_0=1$, $B_1=1$, respectively, the probabilistic paradox \eqref{original Hardy paradox} translates into the logical form:
	\begin{equation}
		\label{original Hardy paradox2}
		\begin{aligned}
			&e_1 = \neg(\neg a_0 \land \neg b_0),\quad &e_2&=\neg(a_0 \land b_1), \\
			&e_3 = \neg(a_1 \land b_0),\quad &e_4&=a_1 \land b_1.
		\end{aligned}
	\end{equation}
	Then the paradox is captured by the implication $e_1 \land e_2 \land e_3 \leq \neg e_4$, which holds in any classical system. Therefore, under the premises $e_1,e_2,e_3$, the occurrence of $e_4$ witnesses quantum contextuality. This pattern, where local constraints imply a global conclusion that is violated by quantum mechanics, is often termed a \emph{failure of transitivity of implication (FTI)} \cite{Stapp1982Mind,Liang2011Specker} and represents the most common form of Hardy-type paradoxes.
	
	A more general formulation of Hardy-type paradoxes avoids the need to distinguish a priori between constraints and conclusions. In classical logic, for events $e$ and $f$, we say $e$ implies $\neg f$ (written $e \leq \neg f$) precisely when $e$ and $f$ cannot occur simultaneously, i.e., when $e \land f = \bot$, where $\bot$ denotes the impossible event.
	
	Thus, a Hardy-type paradox manifests as the observable joint occurrence of $e$ and $f$, despite their classical mutual exclusivity. This idea is also discussed in \cite{Santos2021Conditions}. The following lemma formalizes the equivalence.
	
	\begin{lemma}\label{lem-implication}
		Let $\mathcal{B}$ be a Boolean algebra and $e, f \in\mathcal{B}$. Then $e \leq \neg f$ iff $e \land f = \bot$.
	\end{lemma}
	\begin{proof}
		($\Rightarrow$) If $e \leq \neg f$, then $e \land f \leq \neg f \land f = \bot$, so $e \land f = \bot$. \\
		($\Leftarrow$) If $e \land f = \bot$, then $e = e \land (\neg f \lor f) = (e \land \neg f) \lor \bot = e \land \neg f$, hence $e \leq \neg f$.
	\end{proof}
	More generally, for $e_1,\dots,e_n \in \mathcal{B}$, we have $e_1 \land \cdots \land e_n \leq \neg f$ iff $e_1 \land \cdots \land e_n \land f = \bot$.
	
	Although a classical contradiction $e \land f$ would directly witness contextuality, such joint events are often incompatible in quantum mechanics and therefore unobservable. Hence, demonstrating Hardy-type contextuality typically requires that all but one of the quantum events occur with certainty.
	
	For example, the original Hardy paradox \cite{Hardy1992Quantum,Hardy1993Nonlocality} in logical form is given by:
	\begin{equation}
		\label{Hardy paradox}
		\begin{aligned}
			&e_1 = \neg(a_0 \land b_0),\quad &e_2& = \neg a_1 \to b_0, \\
			&e_3 = \neg b_1 \to a_0,\quad &e_4& = \neg a_1 \land \neg b_1,
		\end{aligned}
	\end{equation}
	where $x \to y := \neg x \lor y$. One can verify that $e_1 \land e_2 \land e_3 \land e_4 = \bot$. This paradox is witnessed by the quantum system in Hardy's state $\ket{\Psi}_{\textup{Hardy}}$:
	\begin{align}
		\ket{\Psi}_{\textup{Hardy}}
		&= N\left( AB\ket{a_0}\!\ket{\neg b_0} + AB\ket{\neg a_0}\!\ket{b_0} + B^2\ket{\neg a_0}\!\ket{\neg b_0} \right)\tag{\textup{H1}} \\
		&= N\left( \ket{a_1}(A\ket{b_0} + B\ket{\neg b_0}) - A^2(A^*\ket{a_1} - B\ket{\neg a_1})\ket{b_0} \right)\tag{\textup{H2}} \\
		&= N\left( (A\ket{a_0} + B\ket{\neg a_0})\ket{b_1} - A^2\ket{a_0}(A^*\ket{b_1} - B\ket{\neg b_1}) \right)\tag{\textup{H3}} \\
		&= N\left( \ket{a_1}\!\ket{b_1} - A^2(A^*\ket{a_1} - B\ket{\neg a_1})(A^*\ket{b_1} - B\ket{\neg b_1})\right)\tag{\textup{H4}}
	\end{align}
	where $N, A, B \in \mathbb{C}$ are complex coefficients, and $\{\ket{a_i}, \ket{\neg a_i}\}$, $\{\ket{b_i}, \ket{\neg b_i}\}$ are orthogonal bases for observables $A_i$ and $B_i$. Let $P$ be the probability induced by $\ket{\Psi}_{\textup{Hardy}}$. Then:
	\begin{align*}
		\textup{From (H1) }:&\ P(a_0\land b_0)=0, i.e., P(\neg(a_0\land b_0))=P(e_1)=1.\\
		\textup{From (H2) }:&\ P(\neg a_1\to b_0)=P(e_2)=1.\\
		\textup{From (H3) }:&\ P(\neg b_1\to a_0)=P(e_3)=1.\\
		\textup{From (H4) }:&\ P(\neg a_1\land\neg b_1)=P(e_4)=\lvert NA^2B^2\rvert^2.
	\end{align*}
	The success probability of the paradox in \eqref{Hardy paradox} is $\mathrm{SP}=P(e_4)=|NA^2B^2|^2$, with a maximum $\mathrm{SP}_{\max}\approx 9\%$ \cite{Hardy1993Nonlocality}.
	
	In fact, any of the four events $e_1, e_2, e_3, e_4$ in \eqref{Hardy paradox} can serve as the “conclusion.” If a quantum state $\rho$ yields $P(e_i) > 0$ for some $i$ while $P(e_k) = 0$ for all $k \ne i$, the contradiction $e_1 \land e_2 \land e_3 \land e_4$ is observed, providing an inequality-free proof of contextuality. This motivates a more general definition.
	
	We now formalize the notion of a logical Hardy-type paradox within the framework of general systems. Let $\mathcal{A}$ be a finite epBA with the classical embedding $i_{\mathcal{A}}:\mathcal{A}\to \mathcal{A}^c$. For an event $e\in\mathcal{A}$, write $e^c:=i_{\mathcal{A}}(e)$. A Hardy-type paradox corresponds to a propositional formula $f$ built from classical events $e^c_1, \dots, e^c_n$ such that $f(e^c_1, \dots, e^c_n) = \bot$. Since any such $f$ is equivalent to a disjunctive normal form, we may assume
	\[
	f = \bigvee_{i=1}^m \bigwedge_{j=1}^k E_{ij} = \bot,
	\]
	where each $E_{ij}$ is either $e^c$ or $\neg e^c$ for some $e \in\mathcal{A}$. The identity $f=\bot$ forces every conjunctive clause $\bigwedge_j E_{ij}$ to be $\bot$. Observing the contradiction $f$ means one of these clauses is realized. Hence, the general form reduces to a single conjunction:
	\[
	\bigwedge_{i=1}^n e^c_i = \bot, \qquad e_i \in\mathcal{A}.
	\]
	
	This leads to the following definition.
	
	\begin{definition}\label{def-LHp}
		Let $\mathcal{A}$ be a finite epBA with the classical embedding and let $p\in s(\mathcal{A})$. A \textbf{logical Hardy-type paradox} on $(\mathcal{A},p)$ is a set of events $\{e_1, \dots, e_n\}\subseteq\mathcal{A}$ satisfying:
		\begin{enumerate}
			\item $e_1^c \land \cdots \land e_n^c = \bot$;
			\item $p(e_k) > 0$ for exactly one $k\in\{1,\dots,n\}$, and $p(e_i) = 1$ for all $i \ne k$.
		\end{enumerate}
		The probability $p(e_k)$ is called the \textbf{success probability (SP)} of the paradox.
	\end{definition}
	
	\section{Logical Contextuality}\label{sec-LC}
	
	Identifying all possible logical Hardy-type paradoxes on a given scenario is challenging. This section presents a result that simplifies the problem. For $e \in \mathcal{A}$, $\lambda \in s_d(\mathcal{A})$, and $p \in s(\mathcal{A})$, the following basic relations are useful:
	\begin{equation}
		\lambda(e) = 1 \;\Leftrightarrow\; \lambda \in e^c; \qquad
		p(e) = 0 \;\Leftrightarrow\; p(\neg e) = 1.
	\end{equation}
	
	Within the sheaf-theoretic framework, Abramsky et al. introduced a hierarchy of contextuality: probabilistic, logical, and strong \cite{Abramsky2011sheaf}. Logical contextuality relaxes probabilistic contextuality by moving from probabilistic to possibilistic constraints. To formalize it in the epBA framework, we first define the possibilistic collapse.
	
	\begin{definition}
		Let $(\mathcal{A},p)$ be a general system. The \textbf{possibilistic collapse} of $p$ is the map $\bar{p}: \mathcal{A} \to \{0,1\}$ given by
		\[
		\bar{p}(x) = 
		\begin{cases}
			0, & p(x) = 0,\\
			1, & p(x) > 0.
		\end{cases}
		\]
	\end{definition}
	
	In sheaf-theoretic language, $\bar{p}$ corresponds to a global section over the Boolean algebra $\{0,1\}$. Logical contextuality is defined by the absence of such a global section compatible with the observed possibilities.
	
	\begin{definition}\label{def-LC}
		A general system $(\mathcal{A}, p)$ is \textbf{logically contextual} if there exists no $p_{\mathcal{A}^c} \in s(\mathcal{A}^c)$ satisfying $\overline{p_{\mathcal{A}^c}}(e^c) = \bar{p}(e)$ for all $e \in \mathcal{A}$.
	\end{definition}
	
	To connect with Hardy-type paradoxes, we need an alternative characterization of logical contextuality, equivalent to Definition~\ref{def-LC} (see \cite{Silva2017Graph,Santos2021Conditions}). We state it as follows.
	
	\begin{theorem}\label{thm-LC2}
		A general system $(\mathcal{A}, p)$ is logically contextual iff there exists an event $e \in \mathcal{A}$ with $p(e) > 0$ such that for every deterministic state $\lambda \in e^c$, one can find an event $e_{\lambda} \in \mathcal{A}$ satisfying $\lambda(e_{\lambda}) = 1$ and $p(e_{\lambda}) = 0$.
	\end{theorem}
	\begin{proof}
		($\Rightarrow$, by contradiction) Assume $(\mathcal{A}, p)$ is logically contextual but the condition fails. Then for every $e \in \mathcal{A}$ with $p(e) > 0$, there exists some $\lambda \in e^c$ such that for all $f \in \mathcal{A}$, $\lambda(f)=0$ or $p(f) > 0$. Define
		\[
		\Lambda = \{\lambda \in s_d(\mathcal{A}) : \forall f \in \mathcal{A},\ \lambda(f)=0 \text{ or } p(f) > 0\}.
		\]
		By assumption, $\Lambda \neq \emptyset$. Let $p'$ be the possibilistic collapse of the state $p_{\mathcal{A}^c} \in s(\mathcal{A}^c)$ given by $p_{\mathcal{A}^c}(\lambda) = 1/|\Lambda|$ for $\lambda \in \Lambda$ and $0$ otherwise. For any $f \in \mathcal{A}$:
		\begin{itemize}
			\item If $p(f) > 0$, then by construction $\exists \lambda \in f^c \cap \Lambda$, so $p'(f^c)=1$.
			\item If $p(f) = 0$, then for any $\lambda \in f^c$ we have $\lambda(f)=1$ but $p(f)=0$, so $\lambda \notin \Lambda$; hence $f^c \cap \Lambda = \emptyset$ and $p'(f^c)=0$.
		\end{itemize}
		Thus $p'(f^c) = \bar{p}(f)$ for all $f$, contradicting logical contextuality.
		
		($\Leftarrow$, by contradiction) Assume the condition holds. Then there exists a set of events $\{e\}\cup\{e_{\lambda} : \lambda\in e^c\}\subseteq\mathcal{A}$ such that:
			\begin{align*}
				p(e) &> 0, \\
				p(e_{\lambda}) &= 0 \quad\text{for all } \lambda\in e^c, \\
				\lambda(e_{\lambda}) &= 1 \quad\text{for all } \lambda\in e^c.
			\end{align*}

	    If $(\mathcal{A}, p)$ is not logically contextual, then there exists $p_{\mathcal{A}^c} \in s(\mathcal{A}^c)$ with $\overline{p_{\mathcal{A}^c}}(f^c) = \bar{p}(f)$ for all $f \in \mathcal{A}$. Take $e \in \mathcal{A}$ as in the condition. For any $\lambda \in e^c$, we have $\lambda \in e_{\lambda}^c$ and $\bar{p}(e_{\lambda}) = 0$, so $\overline{p_{\mathcal{A}^c}}(e_{\lambda}^c)=0$, implying $p_{\mathcal{A}^c}(\lambda)=0$. Hence $p_{\mathcal{A}^c}(e^c)=0$, which gives $\overline{p_{\mathcal{A}^c}}(e^c)=0$. But $\bar{p}(e)=1$ because $p(e) > 0$, a contradiction.
	\end{proof}
	
	Within the sheaf-theoretic approach, Santos et al. showed that for simple $n$-cycle scenarios, possibilistic paradoxes are equivalent to logical contextuality \cite{Santos2021Conditions}. For logical Hardy-type paradoxes, this equivalence extends to arbitrary general scenarios, as stated below.
	
	\begin{theorem}\label{thm-LCiffLHP}
		A general system $(\mathcal{A}, p)$ is logically contextual iff it witnesses a logical Hardy-type paradox.
	\end{theorem}
	\begin{proof}
		($\Rightarrow$) Assume $(\mathcal{A}, p)$ is logically contextual. By Theorem~\ref{thm-LC2}, there exists an event $e \in \mathcal{A}$ with $p(e) > 0$ such that for every $\lambda \in e^c$, one can choose an event $e_{\lambda} \in \mathcal{A}$ satisfying $\lambda \in e_{\lambda}^c$ and $p(\neg e_{\lambda}) = 1$. We claim that $\{e\} \cup \{\neg e_{\lambda} : \lambda \in e^c\}$ forms a logical Hardy-type paradox. Indeed,
		\[
		e^c \land \bigwedge_{\lambda \in e^c} \neg e_{\lambda}^c = \bot,
		\]
		because for each $\lambda \in e^c$ we have $\lambda \notin \neg e_{\lambda}^c$ (since $\lambda \in e_{\lambda}^c$, so the intersection is empty.
		
		($\Leftarrow$) Suppose $(\mathcal{A}, p)$ admits a logical Hardy-type paradox $\{e, \neg e_1, \dots, \neg e_n\}$ with $p(e) > 0$, $p(\neg e_i)=1$ for all $i$, and $e^c \land \bigwedge_{i=1}^n \neg e_i^c = \bot$. By De Morgan, $\bigwedge_i \neg e_i^c = \neg \bigvee_i e_i^c$, so $e^c \subseteq \bigvee_i e_i^c$. Hence, for every $\lambda \in e^c$, there exists $i$ such that $\lambda \in e_i^c$, i.e., $\lambda(e_i)=1$. But $p(e_i)=0$ because $p(\neg e_i)=1$. This fulfills the condition of Theorem~\ref{thm-LC2}, so $(\mathcal{A}, p)$ is logically contextual.
	\end{proof}
	
	Theorem~\ref{thm-LCiffLHP} generalizes earlier results on $(2,k,2)$, $(2,2,d)$ Bell scenarios \cite{Mansfield2012Hardy} and $n$-cycle scenarios \cite{Santos2021Conditions} to any general scenario, including all quantum scenarios.
	
	\subsection{A Logical Hardy-Type Paradox on the Mansfield Scenario}\label{subsec-MF}
	
	Mansfield et al. exhibited a logically contextual state $p_{M}$ in the $(2,3,3)$ scenario that does not yield a coarse‑grained Hardy‑type paradox \cite{Mansfield2012Hardy, Mansfield2013The}. Nevertheless, Theorem~\ref{thm-LCiffLHP} implies that a logical Hardy‑type paradox must exist for $p_{M}$.
	
	The Mansfield scenario is a three‑dimensional bipartite quantum scenario $\mathcal{Q}_M$ that extends naturally to the full $(2,3,3)$ Bell scenario. Its possibilistic collapse $\bar{p}_M$ is given in Table~\ref{table-MF}. In this setup, Alice measures three dichotomic observables, while Bob measures one dichotomic and one trichotomic observable.
	\begin{table}[h]
		\centering
		\scalebox{1.2}{\begin{tabular}{c|cc|ccc|}
				& $b_1$ & $\neg b_1$ & $b_{21}$ & $b_{22}$ & $b_{23}$ \\
				\hline
				$a_1$      & $\mathbf{1}$ & 1 & $\mathbf{0}$ & 1          & 1 \\
				$\neg a_1$ & 1          & 1 & 1          & 1          & 1 \\
				\hline
				$a_2$      & $\mathbf{0}$ & 1 & 1          & 1          & 1 \\
				$\neg a_2$ & 1          & 1 & 1          & $\mathbf{0}$ & 1 \\
				\hline
				$a_3$      & $\mathbf{0}$ & 1 & 1          & 1          & 1 \\
				$\neg a_3$ & 1          & 1 & 1          & 1          & $\mathbf{0}$ \\
				\hline
		\end{tabular}}
		\caption{Possibilistic collapse of $p_{M}$, denoted by $\bar{p}_{M}$. Entries marked $1$ indicate positive probability ($P > 0$), while $0$ denotes impossible events ($P = 0$).}
		\label{table-MF}
	\end{table}
	
	To prove $p_M$ is logically contextual, consider the six bold entries in Table~\ref{table-MF}. For any deterministic state $\lambda$ with $\lambda(a_1 \land b_1)=1$, we have $\lambda(\neg a_1)=\lambda(\neg b_1)=0$; hence all entries in the second row and second column of the table are $0$. To avoid contradicting $\bar{p}_M(a_1 \land b_{21})=0$, we must have $\lambda(b_{22})=1$ or $\lambda(b_{23})=1$. If $\lambda(b_{22})=1$, then $\lambda(a_2)=1$ (because $\bar p_M(\neg a_2\land b_{22})=0$), which forces $\lambda(a_2 \land b_1)=1$, contradicting $\bar{p}_M(a_2 \land b_1)=0$. Similarly, if $\lambda(b_{23})=1$, then $\lambda(a_3)=1$, leading to $\lambda(a_3 \land b_1)=1$ and contradicting $\bar{p}_M(a_3 \land b_1)=0$. Thus, for every such $\lambda$, there exists an event $e$ with $\lambda(e)=1$ but $\bar p_M(e)=0$; therefore $(\mathcal{Q}_M,p_M)$ is logically contextual.
	
	Following the proof of Theorem~\ref{thm-LCiffLHP}, the six events that witness logical contextuality form a logical Hardy‑type paradox:
	\begin{equation}
		\label{Mansfield paradox}
		\begin{aligned}
			&e_1=a_1\land b_1,\quad &e_2&=\neg(a_1 \land b_{21}),\quad &e_3&=\neg(a_2\land b_1)\\
			&e_4=\neg(\neg a_2 \land b_{22}),\quad &e_5&=\neg(a_3 \land b_1),\quad &e_6&=\neg(\neg a_3\land b_{23}).
		\end{aligned}
	\end{equation}
	These satisfy $p_M(e_1) > 0$, $p_M(e_k)=1$ for $k=2,\dots,6$, and classically we have $\bigwedge_{i=1}^6 e_i^c = \bot$. Hence $(\mathcal{Q}_M,p_M)$ witnesses a logical Hardy‑type paradox.
	
	This paradox can be stated as a FTI argument: the conditions $p(e_k)=1$ ($k=2,\dots,6$) force $p(e_1)=0$ for any classical probability $p$. Indeed, if $e_1$ occurred, then $\neg a_2$ and $\neg a_3$ would occur (since $p(a_2\land b_1)=p(a_3\land b_1)=0$), and either $b_{22}$ or $b_{23}$ would occur (because $p(a_1\land b_{21})=0$). But $b_{22}$ contradicts $p(\neg a_2\land b_{22})=0$, and $b_{23}$ contradicts $p(\neg a_3\land b_{23})=0$. Hence $e_1$ cannot occur, so $p(e_1)=0$.
	
	In fact, any inequality‑free contextuality proof based on a logical contradiction can be described as a logical Hardy‑type paradox, including FTI‑type Hardy paradoxes.
	
	\subsection{Strong Contextuality}\label{subsec-SC}
	
	The strongest form in the contextuality hierarchy introduced by sheaf-theoretic approach is strong contextuality \cite{Abramsky2011sheaf}, a special case of logical contextuality defined as follows.
	
	\begin{definition}
		A general system $(\mathcal{A}, p)$ is \textbf{strongly contextual} if for every deterministic state $\lambda \in s_d(\mathcal{A})$, there exists an event $e_{\lambda} \in \mathcal{A}$ such that $\lambda(e_{\lambda}) = 1$ and $p(e_{\lambda}) = 0$.
	\end{definition}
	
	By Theorem~\ref{thm-LC2}, strong contextuality implies logical contextuality. Its connection to logical Hardy-type paradoxes is given by the following theorem, whose proof parallels that of Theorem~\ref{thm-LCiffLHP}.
	
	\begin{theorem}\label{thm-SCiffLHP}
		A general system $(\mathcal{A}, p)$ is strongly contextual iff it witnesses a logical Hardy-type paradox with success probability $\mathrm{SP}=1$.
	\end{theorem}
	\begin{proof}
		($\Rightarrow$) Let $(\mathcal{A}, p)$ be strongly contextual. For each $\lambda \in s_d(\mathcal{A})$, choose $e_{\lambda} \in \mathcal{A}$ with $\lambda \in e_{\lambda}^c$ and $p(\neg e_{\lambda}) = 1$. Then $\{\neg e_{\lambda} : \lambda \in s_d(\mathcal{A})\}$ is a logical Hardy‑type paradox because \[\bigwedge\limits_{\lambda\in s_d(\mathcal{A})} \neg e_{\lambda}^c = \bot;\]
		for any $\lambda$, we have $\lambda \notin \neg e_{\lambda}^c$ (since $\lambda \in e_{\lambda}^c$), so the intersection is empty.
		
		($\Leftarrow$) Suppose $(\mathcal{A}, p)$ admits a logical Hardy‑type paradox $\{\neg e_1, \dots, \neg e_n\}$ with $p(\neg e_i)=1$ for all $i$ and $\bigwedge_i \neg e_i^c = \bot$. By De Morgan, $\bigvee_i e_i^c = s_d(\mathcal{A})$. Hence for every $\lambda \in s_d(\mathcal{A})$, there exists $i$ such that $\lambda \in e_i^c$, i.e., $\lambda(e_i)=1$. Since $p(\neg e_i)=1$, we have $p(e_i)=0$, satisfying the definition of strong contextuality.
	\end{proof}
	
	Well‑known examples of strong contextuality include the Greenberger-Horne-Zeilinger (GHZ) state in the $(3,2,2)$ scenario \cite{Greenberger1989Going} and the Popescu-Rohrlich (PR) box in the $(2,2,2)$ scenario \cite{Popescu1994Quantum}. Moreover, any state on a Kochen‑Specker scenario (where $s_d(\mathcal{A}) = \emptyset$) is trivially strongly contextual \cite{Kochen1967The,Cabello1997bell}.
	
	\section{Incidence Matrix and Atom Graph}\label{sec-IMandAT}
	
	Theorem~\ref{thm-LCiffLHP} provides a systematic way to find Hardy‑type paradoxes via logical contextuality. To implement this, we use incidence matrices developed in the sheaf‑theoretic approach \cite{Abramsky2011sheaf}, which encode the relation between deterministic states (global sections) and atoms (local sections) of a scenario.
	
	For a finite general system $(\mathcal{A}, p)$, Theorem~\ref{thm:atom_graph_correspondence} implies that $(\mathcal{A}, p)$ is fully determined by its atom graph $\mathcal{G}_a(\mathcal{A})$ and the induced state. Deterministic states on $\mathcal{A}$ correspond bijectively to deterministic states on $\mathcal{G}_a(\mathcal{A})$, i.e., functions $\lambda: \mathrm{At}(\mathcal{A}) \to \{0,1\}$ that assign $1$ to exactly one vertex in each maximal clique.
	
	Let $\{\lambda_1,\dots,\lambda_m\}$ be the deterministic states and $\{v_1,\dots,v_n\}$ the vertices (atoms) of $\mathcal{G}_a(\mathcal{A})$. The incidence matrix of $\mathcal{A}$ is the $n\times m$ matrix
	\begin{equation}
	\mathrm{M}(\mathcal{A})[i,j] = 
	\begin{cases}
		0, & \lambda_j(v_i)=0,\\
		1, & \lambda_j(v_i)=1,
	\end{cases}
	\end{equation}
	so the $j$th column is the characteristic vector of $\lambda_j$.
	
	As an illustration, consider the $(2,2,2)$ Bell scenario $\mathcal{Q}_{(2,2,2)}$ generated by events $\{a_0,a_1,b_0,b_1\}$ (see Section~\ref{sec-LHp}). Its 16 atoms are
	\[
	\mathrm{At}(\mathcal{Q}_{(2,2,2)}) = \{a_i\land b_j,\; a_i\land\neg b_j,\; \neg a_i\land b_j,\; \neg a_i\land\neg b_j\}_{i,j=0,1},
	\]
	and its atom graph $\mathcal{G}_a(\mathcal{Q}_{(2,2,2)})$ is shown in Fig.~\ref{fig-atom_graph_222}.
	
	\begin{figure}[H]
		\centering
		\includegraphics[width=0.5\linewidth]{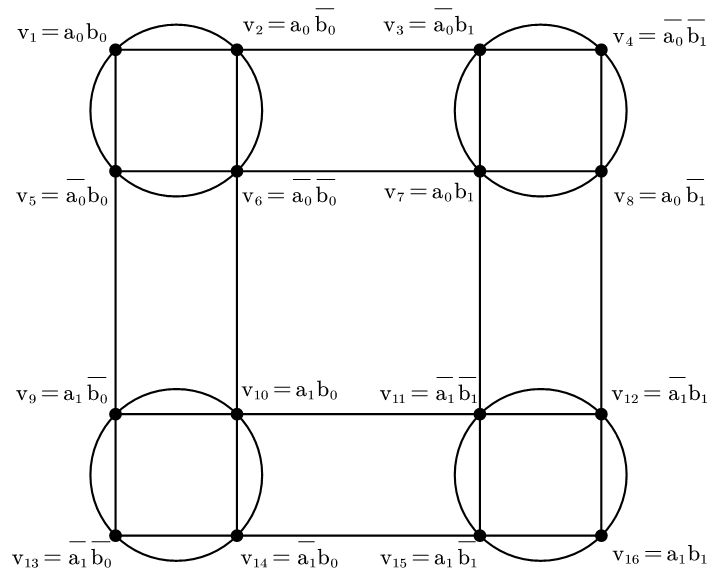}
		\caption{Atom graph $\mathcal{G}_a(\mathcal{Q}_{(2,2,2)})$. Here $\overline{a_i}=\neg a_i$, $\overline{b_j}=\neg b_j$, $a_ib_j=a_i\land b_j$ ($i,j\in\{0,1\}$). Edges denote compatibility, and each straight line or circle is a maximal clique.}
		\label{fig-atom_graph_222}
	\end{figure}
	
	The scenario has 16 deterministic states $s_d(\mathcal{Q}_{(2,2,2)})=\{\lambda_1,\dots,\lambda_{16}\}$; its incidence matrix is:
	\begin{equation}
		\label{eq-M-CHSH}
		\begin{aligned}
			\mathrm{M}(\mathcal{Q}_{(2,2,2)})=&\left(
			\scalebox{0.85}{\begin{tabular}{cccccccccccccccc}
					1 & 1 & 1 & 1 & 0 & 0 & 0 & 0 & 0 & 0 & 0 & 0 & 0 & 0 & 0 & 0 \\
					0 & 0 & 0 & 0 & 1 & 1 & 1 & 1 & 0 & 0 & 0 & 0 & 0 & 0 & 0 & 0 \\
					0 & 0 & 0 & 0 & 0 & 0 & 0 & 0 & 1 & 1 & 1 & 1 & 0 & 0 & 0 & 0 \\
					0 & 0 & 0 & 0 & 0 & 0 & 0 & 0 & 0 & 0 & 0 & 0 & 1 & 1 & 1 & 1 \\
					0 & 0 & 0 & 0 & 0 & 0 & 0 & 0 & 1 & 1 & 0 & 0 & 1 & 1 & 0 & 0 \\
					0 & 0 & 0 & 0 & 0 & 0 & 0 & 0 & 0 & 0 & 1 & 1 & 0 & 0 & 1 & 1 \\
					1 & 1 & 0 & 0 & 1 & 1 & 0 & 0 & 0 & 0 & 0 & 0 & 0 & 0 & 0 & 0 \\
					0 & 0 & 1 & 1 & 0 & 0 & 1 & 1 & 0 & 0 & 0 & 0 & 0 & 0 & 0 & 0 \\
					0 & 0 & 0 & 0 & 1 & 0 & 1 & 0 & 0 & 0 & 1 & 0 & 0 & 0 & 1 & 0 \\
					1 & 0 & 1 & 0 & 0 & 0 & 0 & 0 & 1 & 0 & 0 & 0 & 1 & 0 & 0 & 0 \\
					0 & 0 & 0 & 1 & 0 & 0 & 0 & 1 & 0 & 0 & 0 & 0 & 0 & 1 & 0 & 1 \\
					0 & 1 & 0 & 0 & 0 & 1 & 0 & 0 & 0 & 1 & 0 & 1 & 0 & 0 & 0 & 0 \\
					0 & 0 & 0 & 0 & 0 & 1 & 0 & 1 & 0 & 0 & 0 & 1 & 0 & 0 & 0 & 1 \\
					0 & 1 & 0 & 1 & 0 & 0 & 0 & 0 & 0 & 1 & 0 & 0 & 0 & 1 & 0 & 0 \\
					0 & 0 & 1 & 0 & 0 & 0 & 1 & 0 & 0 & 0 & 0 & 0 & 1 & 0 & 1 & 0 \\
					1 & 0 & 0 & 0 & 1 & 0 & 0 & 0 & 1 & 0 & 1 & 0 & 0 & 0 & 0 & 0
			\end{tabular}}
			\right)_{16\times16}
		\end{aligned}
	\end{equation}
	
	As a non‑Bell example, consider the KCBS scenario $\mathcal{Q}_{\mathrm{KCBS}}$ \cite{Alexander2008Simple}, generated by five rank‑1 projectors $\{\hat P_i\}_{i=0}^4$ on a three‑dimensional Hilbert space $\mathcal{H}$ with $\hat P_i \perp \hat P_{i+1}$ (indices mod 5). Its atom graph $\mathcal{G}_a(\mathcal{Q}_{\mathrm{KCBS}})$ is shown in Fig.~\ref{fig-atom_graph_kcbs}.
	
	\begin{figure}[H]
		\centering
		\includegraphics[width=0.5\linewidth]{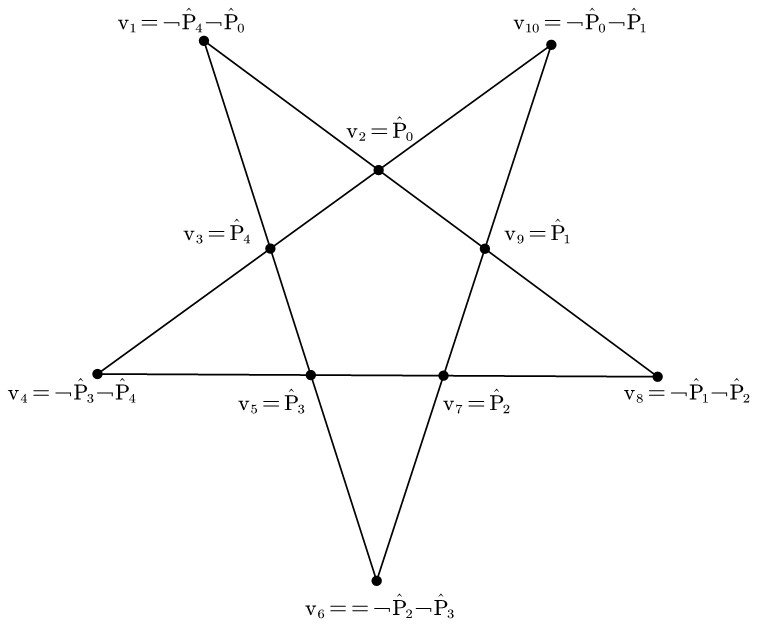}
		\caption{Atom graph $\mathcal{G}_a(\mathcal{Q}_{\mathrm{KCBS}})$. Notation: $\neg\hat P\neg\hat Q = (\neg\hat P)\land(\neg\hat Q)$. Edges indicate compatibility.}
		\label{fig-atom_graph_kcbs}
	\end{figure}
	
	The graph has 10 vertices. By enumeration, $\mathcal{Q}_{\mathrm{KCBS}}$ admits exactly 11 deterministic states, whose incidence matrix is:
	\begin{equation}
		\label{eq-M-KCBS}
		\begin{aligned}
			\mathrm{M}(\mathcal{Q}_{\mathrm{KCBS}})=&\left(
			\begin{tabular}{ccccccccccc}
				1 & 0 & 0 & 1 & 1 & 1 & 0 & 0 & 0 & 0 & 1 \\
				0 & 1 & 0 & 0 & 0 & 0 & 1 & 1 & 0 & 0 & 0 \\
				0 & 0 & 1 & 0 & 0 & 0 & 0 & 0 & 1 & 1 & 0 \\
				1 & 1 & 0 & 0 & 1 & 1 & 0 & 1 & 0 & 0 & 0 \\
				0 & 0 & 0 & 1 & 0 & 0 & 1 & 0 & 0 & 0 & 1 \\
				1 & 1 & 1 & 0 & 0 & 1 & 0 & 0 & 0 & 1 & 0 \\
				0 & 0 & 0 & 0 & 1 & 0 & 0 & 1 & 1 & 0 & 0 \\
				1 & 1 & 1 & 1 & 0 & 0 & 1 & 0 & 0 & 0 & 0 \\
				0 & 0 & 0 & 0 & 0 & 1 & 0 & 0 & 0 & 1 & 1 \\
				1 & 0 & 1 & 1 & 1 & 0 & 0 & 0 & 1 & 0 & 0
			\end{tabular}
			\right)_{10\times11}
		\end{aligned}
	\end{equation}
	
	The incidence matrix yields a Boolean‑system test for logical contextuality based on Definition~\ref{def-LC} \cite{Abramsky2011sheaf}.
	
	\begin{corollary}\label{cor-LCequation}
		Let $\mathcal{A}$ be a finite epBA with atoms $\{v_1,\dots,v_n\}$ and deterministic states $\{\lambda_1,\dots,\lambda_m\}$. A general system $(\mathcal{A},p)$ is logically contextual iff the Boolean equation
		\[
		\mathrm{M}(\mathcal{A})\mathbf{x} = \bar{\boldsymbol{p}}
		\]
		has no solution $\mathbf{x}=(x_1,\dots,x_m)^T\in\mathbb{B}_2^m$, where $\bar{\boldsymbol{p}}=(\bar p(v_1),\dots,\bar p(v_n))^T$.
	\end{corollary}
	\begin{proof}
		($\Rightarrow$) Suppose, for contradiction, that $\mathrm{M}(\mathcal{A})\mathbf{x}=\bar{\boldsymbol{p}}$ has a solution $\mathbf{x}\neq\mathbf0$. Define $p'$ on $s_d(\mathcal{A})$ by $p'(\{\lambda_i\})=x_i$. For any atom $v_i$,
		\[
		p'(v_i^c)=\bigvee_{\substack{j=1\\ \lambda_j(v_i)=1}}^m x_j
		=(\mathrm{M}(\mathcal{A})\mathbf{x})_i
		=\bar p(v_i).
		\]
		Hence $p'(e^c)=\bar p(e)$ for every $e\in\mathcal{A}$, contradicting logical contextuality.
		
		($\Leftarrow$) Assume $(\mathcal{A},p)$ is not logically contextual, so there exists $p_{\mathcal{A}^c}\in s(\mathcal{A}^c)$ with $\overline{p_{\mathcal{A}^c}}(e^c)=\bar p(e)$ for all $e$. Set $x_j=\overline{p_{\mathcal{A}^c}}(\{\lambda_j\})$. Then
		\[
		(\mathrm{M}(\mathcal{A})\mathbf{x})_i
		=\bigvee_{\substack{j=1\\ \lambda_j(v_i)=1}}^m x_j
		=\overline{p_{\mathcal{A}^c}}(v_i^c)
		=\bar p(v_i),
		\]
		so $\mathbf{x}$ solves $\mathrm{M}(\mathcal{A})\mathbf{x}=\bar{\boldsymbol{p}}$, a contradiction.
	\end{proof}
	
	Corollary~\ref{cor-LCequation} gives a Boolean‑equation criterion for logical contextuality of a quantum system $(\mathcal{Q},\rho)$. If $\mathcal{G}_a(\mathcal{Q})$ has $n$ vertices, the set of all possible possibilistic distributions is $\mathbb{B}_2^{\times n}=\{(b_1,\dots,b_n)^T:b_i\in\{0,1\}\}$. An exhaustive search over this set can identify all logically contextual states on $\mathcal{Q}$. For instance, in the $(2,2,2)$ Bell scenario $\mathcal{Q}_{(2,2,2)}$, we find 64 logically contextual Boolean vectors with exactly three zeros. Symmetry of $\mathcal{G}_a(\mathcal{Q}_{(2,2,2)})$ partitions these into 10 equivalence classes, whose representatives are listed below.
	\begin{equation}
		\label{eq-CHSH-10}
		\begin{split}
			\mathbf{b}_1&=(0, 1, 0, 1, 1, 1, 1, 1, 0, 1, 1, 1, 1, 1, 1, 1);\\
			\mathbf{b}_2&=(0, 1, 0, 1, 1, 1, 1, 1, 1, 1, 1, 1, 0, 1, 1, 1);\\
			\mathbf{b}_3&=(0, 1, 1, 0, 1, 1, 1, 1, 1, 1, 1, 1, 0, 1, 1, 1);\\
			\mathbf{b}_4&=(1, 0, 0, 1, 1, 1, 1, 1, 1, 0, 1, 1, 1, 1, 1, 1);\\
			\mathbf{b}_5&=(1, 0, 0, 1, 1, 1, 1, 1, 1, 1, 1, 1, 1, 0, 1, 1);\\
			\mathbf{b}_6&=(1, 0, 1, 0, 1, 1, 1, 1, 1, 0, 1, 1, 1, 1, 1, 1);\\
			\mathbf{b}_7&=(1, 0, 1, 0, 1, 1, 1, 1, 1, 1, 1, 1, 1, 0, 1, 1);\\
			\mathbf{b}_8&=(1, 1, 1, 1, 1, 0, 0, 1, 1, 0, 1, 1, 1, 1, 1, 1);\\
			\mathbf{b}_9&=(1, 1, 1, 1, 1, 0, 1, 0, 1, 0, 1, 1, 1, 1, 1, 1);\\
			\mathbf{b}_{10}&=(1, 1, 1, 1, 1, 0, 1, 0, 1, 1, 1, 1, 1, 0, 1, 1).
		\end{split}
	\end{equation}
	
	The original Hardy paradox \eqref{Hardy paradox} corresponds exactly to the Boolean vector $\mathbf{b}_3$ \cite{Hardy1993Nonlocality}. It is equivalent to the logical Hardy‑type paradox $\{v_{11}, \neg v_1, \neg v_4, \neg v_{13}\}$, where $v_{11}$, $v_1$, $v_4$, $v_{13}$ denote the events $\ket{\neg a_1}\!\ket{\neg b_1}$, $\ket{a_0}\!\ket{b_0}$, $\ket{\neg a_0}\!\ket{\neg b_1}$, $\ket{\neg a_1}\!\ket{\neg b_0}$, respectively. The unique quantum state $\rho$ that witnesses this paradox is the pure state orthogonal to $\ket{a_0}\!\ket{b_0}$, $\ket{\neg a_0}\!\ket{\neg b_1}$ and $\ket{\neg a_1}\!\ket{\neg b_0}$, which is precisely the Hardy state. Its success probability is $\rho(v_{11})=|\bra{\Psi}_{\text{Hardy}}\!\ket{\neg a_1}\!\ket{\neg b_1}|^2$.
	
	\subsection{Logical Hardy‑Type Paradox in the KCBS Scenario}\label{subsec-LHp KCBS}
	
	The KCBS scenario is the simplest $5$-cycle scenario. We find 21 logically contextual Boolean vectors on $\mathcal{Q}_{\mathrm{KCBS}}$, which fall into 5 equivalence classes by the rotational symmetry of $\mathcal{G}_a(\mathcal{Q}_{\mathrm{KCBS}})$. Representative vectors are shown in Fig.~\ref{KCBS-5}.
	\begin{figure}[H]
		\centering
		\includegraphics[width=1\linewidth]{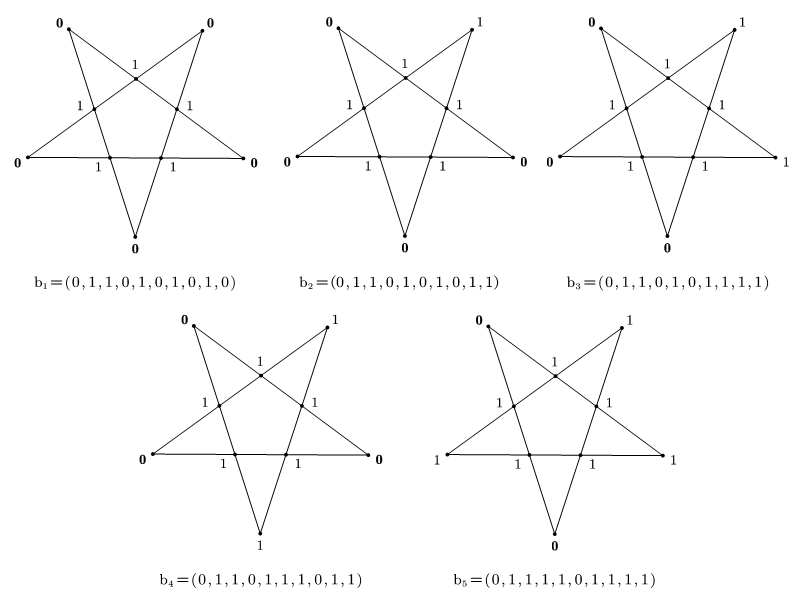}
		\caption{Five Boolean vectors $\{\mathbf{b}_i\}_{i=1}^5$ corresponding to the logically contextual states on $\mathcal{Q}_{\mathrm{KCBS}}$.}
		\label{KCBS-5}
	\end{figure}
	All $\mathbf{b}_i$ are possibilistic collapses of states on $\mathcal{G}_a(\mathcal{Q}_{\mathrm{KCBS}})$, so there are exactly 5 distinct types of logically contextual states. However, not every type admits a quantum realization; we must check which $\mathbf{b}_i$ correspond to possibilistic collapses of quantum states.
	
	Because the KCBS scenario is generated by rank‑1 projectors on a 3‑dimensional Hilbert space, every vertex of $\mathcal{G}_a(\mathcal{Q}_{\mathrm{KCBS}})$ is a rank‑1 projector. Take $\mathbf{b}_4$ (Fig.~\ref{KCBS-5}). If it came from a quantum state $\rho$, then $\rho(v_1)=\rho(v_4)=0$, forcing $\rho$ to be supported on the subspace orthogonal to both $v_1$ and $v_4$. The only such projector is $v_3$, so $\rho=v_3$. But $v_3$ is orthogonal to $v_2$, giving $\rho(v_2)\neq1$, contradicting $\mathbf{b}_4(v_2)=1$. Hence $\mathbf{b}_4$ is not quantum realizable. Similar arguments rule out $\mathbf{b}_1,\mathbf{b}_2,\mathbf{b}_3$.
	
	Now consider $\mathbf{b}_5$. If it arises from a quantum state $\rho$, then $\rho(v_1)=\rho(v_6)=0$, so $\rho=\ket{\psi}\!\bra{\psi}$ must be a pure state orthogonal to both $v_1$ and $v_6$. Using the explicit construction of the KCBS scenario from \cite{Adan2010Contextuality} (with $v_i=\ket{v_i}\!\bra{v_i}$): 
	\begin{align*}
			\ket{v_2}&= (1, 0, \sqrt{\cos(\pi/5)})^T, \\
			\ket{v_3}&= (\cos(4\pi/5), -\sin(4\pi/5), \sqrt{\cos(\pi/5)})^T.\\
			\ket{v_5}&= (\cos(2\pi/5), \sin(2\pi/5), \sqrt{\cos(\pi/5)})^T, \\
			\ket{v_7}&= (\cos(2\pi/5), -\sin(2\pi/5), \sqrt{\cos(\pi/5)})^T, \\
			\ket{v_9}&= (\cos(4\pi/5), \sin(4\pi/5), \sqrt{\cos(\pi/5)})^T.
	\end{align*}
	
	The vectors orthogonal to the required contexts are obtained via the cross product:
	\begin{align*}
		\ket{v_1}&=\ket{v_2}\times\ket{v_3}=(\sqrt{\cos(\pi/5)\sin(\pi/5)},-\sqrt{\cos(\pi/5)}(\cos(\pi/5)+1), -\sin(\pi/5))^T,\\
		\ket{v_6}&=\ket{v_5}\times\ket{v_7}=(2\sin(2\pi/5)\sqrt{\cos(\pi/5)}, 0, -\sin(\pi/5))^T,
	\end{align*}
	
	Let $\ket{\psi}:=\ket{v_1}\times\ket{v_6}$. Numerical computation shows that the possibilistic collapse of $\ket{\psi}$ equals $\mathbf{b}_5$, so $\ket{\psi}$ is the unique logically contextual quantum state in this KCBS realization. From Theorem~\ref{thm-LCiffLHP}, $\{v_9, \neg v_1, \neg v_6\}$ is a logical Hardy‑type paradox on $\mathcal{Q}_{\mathrm{KCBS}}$ induced by $\mathbf{b}_5$ with success probability $p(v_9)$.  
	
	The vertices $v_9$, $v_1$, $v_6$ form a triangle in the atom graph (Fig.~\ref{fig-atom_graph_kcbs}). By rotational symmetry of $\mathcal{G}_a(\mathcal{Q}_{\mathrm{KCBS}})$,  we obtain other four equivalent paradoxes, giving five equivalent types in total.
	\begin{align*}
		\mathrm{H}_1&=\{v_9,\neg v_1,\neg v_6\}, \\
		\mathrm{H}_2&=\{v_2,\neg v_4,\neg v_8\},\\
		\mathrm{H}_3&=\{v_3,\neg v_6,\neg v_{10}\}, \\
		\mathrm{H}_4&=\{v_5,\neg v_1,\neg v_8\}, \\
		\mathrm{H}_5&=\{v_7,\neg v_4,\neg v_{10}\}.
	\end{align*}
	
	For the explicit KCBS realization in \cite{Adan2010Contextuality}, the unique quantum state satisfying $\rho(v_1)=\rho(v_6)=0$ and $\rho(v_9)>0$ is a pure state $\rho=\ket{\psi}\!\bra{\psi}$. Direct calculation yields the success probability $\mathrm{SP}_1=\rho(v_9)\approx10.56\%$; the other four paradoxes give the same value.
	
	While the maximal success probability for Hardy‑type paradoxes in general $5$-cycle quantum scenarios is known to be $1/9\approx11.11\%$ \cite{Santos2021Conditions}, the maximum achievable for logical Hardy‑type paradoxes on the specific KCBS scenario $\mathcal{Q}_{\mathrm{KCBS}}$ remains to be determined.

	\section{Conclusion and Outlook}\label{sec-conclusion}
	
	We introduce a logical formalism for Hardy-type paradoxes within the exclusive partial Boolean algebra (epBA) framework. We prove that a general system exhibits logical contextuality if and only if it witnesses a logical Hardy-type paradox. This result generalizes earlier work restricted to specific scenarios, such as $(2,k,2)$ and $(2,2,d)$ Bell scenarios \cite{Mansfield2012Hardy} and $n$-cycle scenarios \cite{Santos2021Conditions}. In particular, a system is strongly contextual precisely when it admits a logical Hardy-type paradox with success probability $\mathrm{SP} = 1$.
	
	Using incidence matrices \cite{Abramsky2011sheaf} and atom graphs \cite{Liu2025Atom}, we classify 10 types of quantum-observable Hardy-type paradoxes in the $(2,2,2)$ Bell scenario, one of which recovers the original Hardy paradox. We also identify the unique logical Hardy-type paradox in the KCBS scenario $\mathcal{Q}_{\mathrm{KCBS}}$, achieving $\mathrm{SP} \approx 10.56\%$ for a specific parameter setting.
	
	A probabilistic relaxation leads to Cabello-type paradoxes \cite{Cabello2002Bell,Liang2005Nonlocality,Kunkri2006Nonlocality,Chen2024Hardy}, which rely on statistical inequalities rather than the inequality-free condition of Hardy-type paradoxes. Our framework can be extended to encompass Cabello-type paradoxes by a natural modification of Definition~\ref{def-LHp}: we require two events $e_k, e_l \in\{e_i\}_{i=1}^n$ with $p(e_k) > 0$, $p(e_l) > 0$, and the degree of success $p(e_k) - p(\neg e_l) > 0$ (or $p(e_l) - p(\neg e_k) > 0$). This defines a logical Cabello paradox, extending the approach to probabilistic scenarios where the inequality-free condition fails.
	
	This work focuses on scenarios satisfying Specker’s principle, such as ideal measurements. Extending the analysis to general measurements may yield further results on Hardy-type paradoxes.
	
	\bibliography{math}
	
\end{document}